\documentclass[journal]{IEEEtran}
\ifCLASSINFOpdf
\else
\fi
\usepackage{amsmath,amssymb,amsthm}
\usepackage{algorithm}
\usepackage{algpseudocode}

\usepackage{xcolor}

\usepackage{comment}
\usepackage{todonotes} 
\usepackage{xspace}
\presetkeys{todonotes}{inline,fancyline}{}

\newcommand{\NB}{\mathit{NB}}

\theoremstyle{definition}
\newtheorem{definition}{Definition}
\newtheorem{fact}{Fact}

\newtheorem{example}{Example}
\newtheorem{theorem}{Theorem}

\newtheorem{remark}{Remark}

\newtheorem{prop}{Proposition}
\theoremstyle{plain}
\newtheorem{problem}{Problem}

\usepackage{tikz}
\usetikzlibrary{automata,arrows,tikzmark,decorations.pathreplacing}
\usepackage{wonhamautomata}

\begin{document}
%
\title{Supervisory Control Theory with Event Forcing}
%
%
%

\author{Michel~Reniers,~\IEEEmembership{Senior Member,~IEEE,}
        and Kai~Cai,~\IEEEmembership{Senior Member,~IEEE}
\thanks{M. Reniers is with the Department
of Mechanical Engineering, Eindhoven University of Technology, Eindhoven, The Netherlands. E-mail: m.a.reniers@tue.nl.}
\thanks{K. Cai is with Department of Core Informatics, Osaka Metropolitan University, Osaka, Japan. E-mail: cai@omu.ac.jp.}
}

\maketitle

\begin{abstract}
In the Ramadge-Wonham supervisory control theory the only interaction mechanism between supervisor and plant is that the supervisor may enable/disable events from the plant and the plant makes a final decision about which of the enabled events is actually taking place. 
In this paper, the interaction between supervisor and plant is enriched by allowing the supervisor to force specific events (called forcible events) that are allowed to preempt uncontrollable events. 
A notion of {\em forcible-controllability} is defined that captures the interplay between controllability of a supervisor w.r.t.\ the uncontrollable events provided by a plant in the setting with event forcing. 
Existence of a maximally permissive, forcibly-controllable, nonblocking supervisor is shown and an algorithm is provided that computes such a supervisor.
The approach is illustrated by two small case studies.
\end{abstract}

\begin{IEEEkeywords}
Discrete event systems, finite automata, forcible events, forcible controllability supervisory control, nonblocking.
\end{IEEEkeywords}

%
\IEEEpeerreviewmaketitle

\section{Introduction}
\label{section:Introduction}

\IEEEPARstart{T}{he} main interaction mechanism between a supervisor and the plant it controls in supervisory control theory (SCT) \cite{bookcai,cassandras2009introduction} is the mechanism of enabling/disabling events offered by the plant. In such a setting the supervisor indicates which of the events that are enabled in the plant are allowed to occur, but the supervisor does not dictate which of these allowed events will occur. 

One can observe that in many implementations of supervisory controllers a different interaction involving event forcing is used between supervisor and plant.

We mention two important consequences of this `mismatch' between this assumption in SCT and the actual setting in implementations: (1) a proper supervisory controller does not necessarily result in a proper implementation (especially w.r.t.\ the infamous property of nonblocking there are complications), and (2) loss of permissiveness because final choice of executed event is assumed to be in the plant. In \cite{malik2002generating,Reijnen19,Reijnen2022}, the mismatch between a proper supervisory controller and an implementation is discussed in detail.

In this paper we study the second consequence in more detail. We propose a SCT that does not only allow enabling/disabling by the supervisor, but also, for some events, allows forcing of such an event. Forcing of an event results in preemption of other events and may thus contribute to obtaining proper supervisors. 

In \cite{Golaszewski1987}, supervisory control of untimed discrete event systems with forcible events is studied. In \cite[Example 2]{Golaszewski1987}, a version of controllability was introduced  (when interpreted in a setting with uncontrollable events) which requires that a sublanguage is closed under all uncontrollable events available in the plant or allows for a single forced event only. A clear difference with our definition is that we generally allow multiple forcible events to be available for the sake of maximal freedom of choices. In \cite{Golaszewski1987}, this specific condition that only one forcible event is considered an acceptable continuation is not further elaborated outside that example.

In \cite[section 3.8]{bookcai} forcible events are introduced also for a setting without timing. It is argued (informally) that forcible events are a modeling issue by showing how synthesis in the context of forcible events can be achieved by traditional synthesis on a transformed plant.
Drawbacks of this approach are the possible increase of the number of states in the state space on which synthesis is to be performed, and the fact that the transformation as suggested takes place on the state space and not on the individual plant components typically used in modeling an uncontrolled system. 
This paper considers the same setting, but studies a direct synthesis algorithm for the setting of supervisory control with forcible events (thus countering both disadvantages).

In extensions of SCT with notions of timing, such as different types of Timed DES \cite{brandin1994supervisory,ZhaCaiWon13,TakaiUshio06,rashidinejad2018,Rashidinejad2020a} and TA \cite{Rashidinejad2020b}, forcible events are introduced to preempt progress of time (which is generally considered uncontrollable). In all these works, the result is an adapted notion of controllability where progress of time is consider uncontrollable only in cases where no forcible events are enabled. A main difference with our approach is that we consider the status of every event as being either controllable or uncontrollable as fixed and given. Additionally, events may be forcible, which allows them to be used by the supervisor to preempt events (not only progress of time, but also other events).

In \cite{brandin1994supervisory}, a distinction is made between strongly preemptive forcing, where a forcible event preempts any other eligible event, and weakly preemptive forcing, where premption is only assumed w.r.t.\ the passage of time. 
In \cite{ZhaCaiWon13,TakaiUshio06}, weakly preemptive forcing is considered. 
In contract, the notion of preemption adopted in our paper is that of strong preemption with the side note that a supervisor may have multiple forcible event alternatives instead of exactly one.
A notion of controllability is defined that depending on the eligibility of a forcible event treats the event $\mathit{tick}$ as uncontrollable or not (uncontrollable when no forcible event is eligible). 

In \cite{rashidinejad2018}, the supervisory control of a plant is studied where the interaction between supervisor and plant is assumed to result in delays. Time is modeled by means of a tick event, and forcible events are used to preempt this tick event (and no other events). The notion of controllability used in \cite{rashidinejad2018} aligns well with the notion defined in this paper. Although the synthesis algorithm presented in \cite{rashidinejad2018} is quite different from the version in this paper, the underlying idea that a tick event may be preempted by a forcible event is recognizable and also the situation that during synthesis all forcible events may disappear resulting in these states being bad states after all is there.

In both \cite{Rashidinejad2020a} and \cite{Rashidinejad2020b}, forcible events are used to preempt the progress of time in a context of timed automata. In such automata the progress of time is, in contrast with Timed DES, described by real-valued clocks and separated from the occurrence of events.

In \cite{Ushio2005}, forcible events are used in the context of hybrid automata. The treatment is restricted to the setting where all forcible events are also controllable, an assumption that is not adopted in the present paper.

In \cite{Huang2008}, existence and synthesis of safe and nonblocking directed control for discrete event systems is studied. A directed controller selects at most one controllable event to be executed from each system state. Our forcibly-controllable supervisor is possibly more permissive than a directed controller as multiple controllable events may still be allowed in order to maximize permissiveness in the sense of choosing alternative forcible events.

In \cite{Reijnen19}, supervised control is applied. In supervised control a supervisor is meant to be implemented
together with a separate controller. The supervisor
is used to monitor the behavior of the plant and disables some events, whereas the controller forces some of the events enabled by the supervisor to occur in the plant. Forcible events are used in \cite{Reijnen19}, but not for the synthesis of a maximally permissive supervisor, but for the implementation of such a supervisor as a controller.
In all the four case studies mentioned in the paper all controllable events were assumed to be forcible by the controller.

In \cite{Balemi93}, the authors point to three reasons for limited applicability of supervisory control theory. One of these three reasons is that the model interpretation of SCT does not connect well with applications: ``The logical plant model proposed in supervisory control theory assumes a plant that
`generates' events spontaneously unless it is prevented
from doing so. The control mechanism available to the
supervisor is the ability to prevent the occurrence of some
events, called controllable events. This model is not appropriate
for most real systems. In fact, real systems usually
react to commands as inputs with responses as outputs.'' The authors assume an input-output model where the supervisor has to be controllable and the plant has to be complete for commands issued by the supervisor.

In the present paper, we enrich the traditional options for a supervisor to enable/disable events with the option to preempt events by forcing some event to occur. We however, do not assume that the plant is complete for such commands as \cite{Balemi93} does. Instead we formulate a new forcible-controllability property that the supervisor must satisfy.

The contributions of this paper can be summarized as follows:
\begin{enumerate}
\item In this paper we provide a basic untimed supervisory control framework where besides the traditional enabling/disabling of events generated by the plant also event forcing is considered. 
\item We formulate a notion of forcible controllability that captures the interplay between controllability of a supervisor w.r.t.\ the uncontrollable events provided by a plant in the setting where forcing of events may be used to preempt such uncontrollable events. 
\item We provide the basic properties of forcible-controllability and show that a maximally permissive forcibly-controllable supervisor exists. 
\item We provide a synthesis algorithm that computes such a supervisor for a given plant automaton.
\end{enumerate}

\paragraph*{Structure of the paper}
In Section \ref{section:forcingsetting} some preliminaries about supervisory control theory are introduced that are related to the subject material of this paper, and the forcing supervisory control problem is formulated. Then, in Section \ref{section:FCsublanguages}, the notion of forcible-controllability is introduced that is crucial in capturing the interplay between plant and (forcing) supervisor. Properties of this notion are discussed in detail, and it is proved that this notion is necessary and sufficient for the solvability of the formulated forcing supervisory control problem. In Section \ref{section:supremalFC} it is shown that the supremal forcibly-controllable sublanguage of a given language exists, which provides a maximally permissive solution to the forcing supervisory control problem. Section \ref{section:algorithm} presents an algorithm for computing the maximally permissive supervisor and states its properties such as termination and correctness. Sections \ref{section:smallmanufacturingsystem} and \ref{section:smallfactory} illustrate the outcome of the algorithm and showcase the benefit of introducing event forcing. The paper is concluded in Section \ref{section:conclusions}.

\section{Forcible Events and Forcing Supervisory Control}\label{section:forcingsetting}

Consider that a plant to be controlled is modeled by a finite automaton $P = (Q,\Sigma,\delta,q_0,Q_m)$ where $Q$ is a finite set of states, $\Sigma$ is a finite set of events partitioned into controllable events $\Sigma_c$ and uncontrollable events $\Sigma_u$, $\delta : Q \times \Sigma \rightarrow Q$ is a partial transition function, $q_0\in Q$ is the initial state, and $Q_m \subseteq Q$ is the set of marker states. A controllable event $c \in \Sigma_c$ can be disabled by an external agent (supervisor), whereas an uncontrollable event $u \in \Sigma_u$ cannot be disabled. 

Plant $P$ generates a closed language $L(P)$ and a marked language $L_m(P)$, as defined in the usual way \cite{bookcai}.

\begin{definition}[Nonblocking]
A finite automaton $P=(Q,\Sigma,\delta,q_0,Q_m)$ is called {\em nonblocking} if $L(P) = \overline{L_m(P)}$.
\end{definition}

For a string $s \in L(P)$, let $E_P(s) := \{\sigma \in \Sigma \mid s \sigma \in L(P)\}$ be the subset of eligible events that can occur after $s$ in $P$. 

\begin{definition}[Controllable sublanguage]
Given a (nonblocking) plant $P$, a sublanguage $F \subseteq L_m(P)$ is {\em controllable} w.r.t.\ $P$ if  
\[(\forall {s \in \overline{F}},~\forall {u \in \Sigma_u}) ~ [u \in E_P(s) \implies su \in \overline{F}]. \]
\end{definition}

Controllable sublanguages are closed under union (see e.g., \cite{bookcai}). Thus given an arbitrary language $F \subseteq L_m(P)$, there exists the supremal controllable sublanguage of $F$. 

\begin{figure*}[t!]
    \begin{center} \includegraphics[width=13cm]{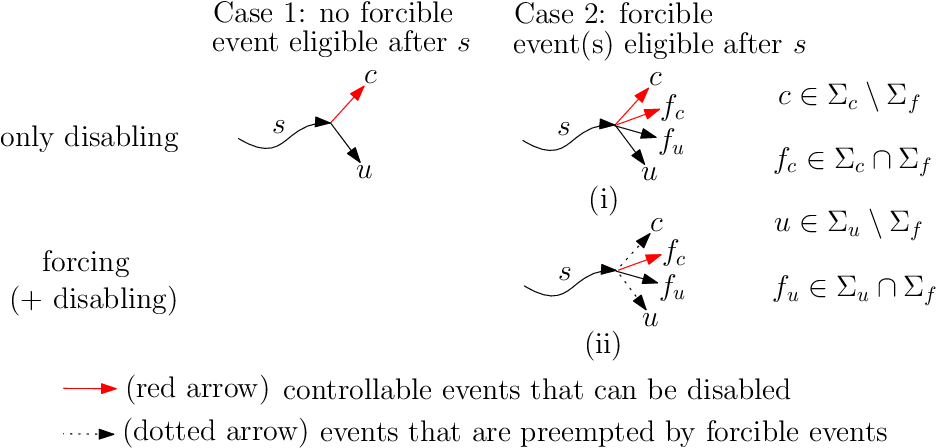}
    \caption{Control mechanisms for supervisory control. \label{fig:mechanism}}
    \end{center}
\end{figure*}

Now bring in a subset of {\em forcible events} $\Sigma_f \subseteq \Sigma$. A forcible event $f \in \Sigma_f$ can be forced by an external agent in order to preempt other event occurrences. Thus forcible events provide a new mechanism for control, which is distinct from controllable events (that can be disabled themselves but cannot preempt other events). 

A forcible event $f\in \Sigma_f$ can be either controllable or uncontrollable. If $f \in \Sigma_f \cap \Sigma_c$, then $f$ can either be forced or be disabled. For example, `cross\_street' can be an event that a pedestrian can force (by accelerating upon a slowly approaching car) or disabled (i.e. giving up crossing upon a fast approaching car). 
On the other hand, if $f \in \Sigma_f \cap \Sigma_u$, then $f$ can be forced but cannot be disabled. An example of such an uncontrollable forcible event is `land\_plane', which may be forced (say within 15 minutes) but cannot be prevented (plane has to land eventually) \cite{bookcai}.

With this added forcing mechanism, we define the corresponding {\em supervisory control} for plant $P$ which is any function $V : L(P) \rightarrow Pwr(\Sigma)$ (here $Pwr(\cdot)$ denotes powerset). Let us analyze the characteristics of this supervisory control $V$. Consider a strings $s  \in L(P)$, and let $c,f_c,u,f_u$ respectively be a nonforcible controllable, forcible controllable, nonforcible uncontrollable, forcible uncontrollable
event. As displayed in Fig.~\ref{fig:mechanism}, we discuss two cases.

Case~1: $E_P(s) \cap \Sigma_f = \varnothing$, i.e. no forcible event is eligible after $s$. In this case, only controllable events may be disabled (e.g. $c$ in Fig.~\ref{fig:mechanism}, Case 1).

Case~2: $E_P(s) \cap \Sigma_f \neq \varnothing$, i.e., there exists at least one forcible event eligible after $s$. In this case, the eligible forcible event(s) may be either forced or not forced. When all the eligible forcible events are not forced (see Fig.~\ref{fig:mechanism}(i), Case 2), again only controllable events may be disabled (e.g., $c, f_c$ in Fig.~\ref{fig:mechanism}(i)). On the other hand, if an eligible forcible event is forced (see Fig.~\ref{fig:mechanism}(ii), Case 2), then all the other eligible events after $s$ are preempted. Since in general there are multiple forcible events that may be forced after $s$ (e.g., $f_c,f_u$ in Fig.~\ref{fig:mechanism}(ii)), we will define the supervisory control function to include all such forcible events for the sake of maximal freedom of choice. In implementation, one of these forcible events will be chosen (say by an external forcing agent) such that other events are preempted; our model will leave such a choice nondeterministic (much like the nondeterministic occurrence of one event among multiple enabled events, e.g., $u,f_u$ in Fig.~\ref{fig:mechanism}(i)). It is also noted that disabling a forcible controllable event is also possible (e.g., $f_c$ in Fig.~\ref{fig:mechanism}(ii)), and such disabled forcible events cannot be forced. On the other hand, a forcible uncontrollable event cannot be disabled (e.g., $f_u$ in Fig.~\ref{fig:mechanism}(ii)), and such forcible events may be either forced or not forced.

With the above discussion on possible scenarios of disabling and forcing control mechanisms, we define the supervisory control $V : L(P) \rightarrow Pwr(\Sigma)$ such that
\begin{align} \label{eq:V}
    V(s) = \begin{cases}
  \Sigma_u \cup \Sigma_c'  & \text{ if } E_P(s) \cap \Sigma_f = \varnothing, \\
  {\rm either}~\Sigma_u \cup \Sigma_c'~{\rm or}~\Sigma_f'& \text{ if } E_P(s) \cap \Sigma_f \neq \varnothing. 
\end{cases}
\end{align}
Here $\Sigma'_c \subseteq \Sigma_c$, $\Sigma'_f \subseteq \Sigma_f$. The first line corresponds to Case~1 above, whereas the second line to Case~2. In the second line, `either $\Sigma_u \cup \Sigma'_c$' is when none of the eligible forcible events are forced (Fig.~\ref{fig:mechanism}(i)), while `or $\Sigma'_f$' is when at least one eligible forcible event is forced (Fig.~\ref{fig:mechanism}(ii)).

Write $V/P$ for the {\em closed-loop system} where the plant $P$ is under the supervisory control of $V$. The closed language generated by $V/P$ is defined as follows:
\begin{enumerate}
    \item [(i)] $\epsilon\in L(V/P)$;
    \item [(ii)] $s\in L(V/P),  \sigma\in V(s), s\sigma\in L(P) \Rightarrow s\sigma\in L(V/P)$;
    \item [(iii)] no other strings belong to $L(V/P)$.
\end{enumerate}
Let $F \subseteq L_m(P)$. The marked language of $V/P$ w.r.t.~$F$ is 
\begin{align*}
L_m(V/P) := L(V/P) \cap F.
\end{align*}
We say that $V/P$ is nonblocking if 
\begin{align*}
L(V/P) = \overline{L_m(V/P)}.
\end{align*}

Now we are ready to formulate the forcing supervisory control problem.

\begin{problem}[Forcing supervisory control problem] \label{problem:fsc}
Given a plant $P$ and a specification $\varnothing \neq F \subseteq L_m(P)$, design a supervisory control $V : L(P) \rightarrow Pwr(\Sigma)$ such that 
\begin{enumerate}
    \item [(i)] $V/G$ is nonblocking;
    \item [(ii)] $L_m(V/G) = F$.
\end{enumerate}
\end{problem}

In words, Problem~\ref{problem:fsc} is after a supervisory control that not only renders the closed-loop system nonblocking but also synthesizes the imposed specification. 

\section{Forcibly-Controllable Sublanguages}
\label{section:FCsublanguages}

\begin{figure*}[t!]
    \begin{center} \includegraphics[width=.7\textwidth]{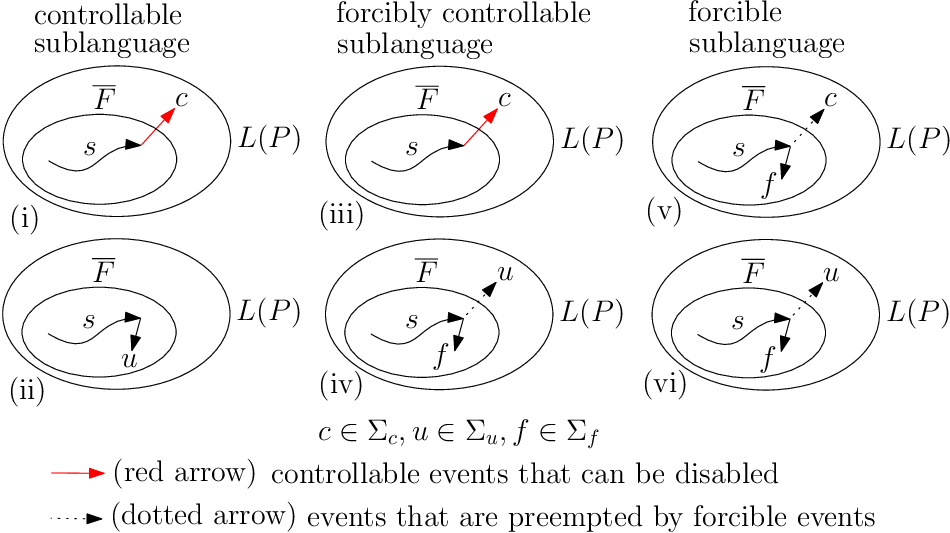}
    \caption{Comparisons among controllable, forcibly-controllable, and forcible sublanguages.\label{fig:confor}}
    \end{center}
\end{figure*}

To solve Problem~\ref{problem:fsc}, the following new language property is key.

\begin{definition}[Forcibly-controllable sublanguage]\label{defn:forcon}
Given a (nonblocking) plant $P$, a sublanguage $F \subseteq L_m(P)$ is {\em forcibly-controllable} w.r.t.\ $P$ if
\[ (\forall s \in \overline{F}) ~ \left[
\begin{array}{@{}l}
(\forall u \in \Sigma_u)~[ u \in E_P(s) \implies su \in \overline{F}] \\
\lor \\
(\exists f \in \Sigma_f)~[sf \in \overline{F}]\land (\forall \sigma \in \Sigma\setminus\Sigma_f)~[s \sigma \not\in \overline{F}] \end{array} \right]. 
\]
\end{definition}

In words, a sublanguage $F$ is forcibly-controllable if either $F$ is controllable or there exists a forcible event that keeps $\overline{F}$ invariant by preempting all non-forcible events. 

Thus by definition, if a sublanguage $F$ is controllable, $F$ is also forcibly-controllable. The difference between these two properties is illustrated in Fig.~\ref{fig:confor}(ii) and (iv). While controllability does not allow uncontrollable event $u$ to take string $s$ out of $\overline{F}$ (Fig.~\ref{fig:confor}(ii)), forcible-controllability allows so if there exists a forcible event $f$ after $s$ to keep $sf$ inside $\overline{F}$ by preempting $u$ (Fig.~\ref{fig:confor}(iv)). Based on Fig.~\ref{fig:confor}(iv), it is readily seen that a forcibly-controllable sublanguage need not be controllable.

Now that we have seen that controllability involves pure disabling, forcible-controllability invovles both disabling and forcing, it is natural to define another language property which involves pure forcing. In the following, we define this language property -- forcible sublanguages -- which is not needed in our main development but interesting in itself as compared to the other language properties.

\begin{definition}[Forcible sublanguage]\label{defn:for}
Given a (nonblocking) plant $P$, a sublanguage $F \subseteq L_m(P)$ is {\em forcible} w.r.t.\ $P$ if
\[ (\forall s \in \overline{F}) ~ \left[
\begin{array}{@{}l}
(\forall \sigma \in \Sigma)~[ \sigma \in E_P(s) \implies s\sigma \in \overline{F}] \\
\lor \\
(\exists f \in \Sigma_f)~[sf \in \overline{F}]\land (\forall \sigma \in \Sigma\setminus\Sigma_f)~[s \sigma \not\in \overline{F}] \end{array} \right]. 
\]
\end{definition}

Compared to Definition~\ref{defn:forcon}, forcible sublanguage is defined by only replacing the universal quantification on the first line with $u \in \Sigma_u$ by $\sigma \in \Sigma$. This change means that whenever a string $s$ is taken out of $\overline{F}$ by any event $\sigma$, controllable or uncontrollable, there must exist at least one forcible event $f$ after $s$ to keep $sf$ inside $\overline{F}$ by preempting $\sigma$. This is illustrated in Fig.~\ref{fig:confor}(v) and (vi). 

Thus by definition, if a sublanguage $F$ is forcible, $F$ is also forcibly-controllable. But the reverse need not be true: a forcibly-controllable sublanguage allows $s$ to exit $\overline{F}$ by a controllable event without the presence of a forcible event (Fig.~\ref{fig:confor}(iii)), which is not allowed by a forcible sublanguage.

In addition, it is not difficult to see that forcible sublanguages and controllable sublanguagess generally are not related. A forcible sublanguage need not be controllable, because the scenario in Fig.~\ref{fig:confor}(vi) is allowed by a forcible sublanguage but not by a controllable sublanguage. Conversely, since the scenario in Fig.~\ref{fig:confor}(i) is allowed by controllable sublanguage but not by forcible sublanguage, a controllable sublanguage need not be forcible.
 
We summarize the above relationships among the three language properties.

\begin{fact}
Let plant $P$ be a finite automaton over $\Sigma = \Sigma_c \cup \Sigma_u$ and a forcible event set $\Sigma_f \subseteq \Sigma$. Also let $F \subseteq L_m(P)$. 
\begin{enumerate}
    \item If $F$ is controllable, it is also forcibly-controllable.
    \item If $F$ is forcible, it is also forcibly-controllable.
    \item There are focibly-controllable sublanguages that are not controllable.
    \item There are forcibly-controllable sublanguages that are not forcible.
    \item There are controllable sublanguages that are not forcible. 
    \item There are forcible sublanguages that are not controllable.
\end{enumerate}
\end{fact}

At this point, it is also convenient to state the following facts about the three language properties.

\begin{fact}
\label{prop2}
Let plant $P$ be a finite automaton over $\Sigma = \Sigma_c \cup \Sigma_u$ and a forcible event set $\Sigma_f \subseteq \Sigma$.
\begin{enumerate}
    \item The empty language $\varnothing$ is controllable, forcibly-controllable, and forcible.
    \item The closed language $L(P)$ is controllable, forcibly-controllable, and forcible. \label{initialcase}
    \item Let $\Sigma_f = \varnothing$. Then $F$ is forcibly-controllable iff $F$ is controllable.
    \item Let $\Sigma_c = \varnothing$. Then $F$ is controllable iff $F$ is forcible iff $F$ is forcibly-controllable. 
    \item Let $\Sigma_u = \varnothing$. Any sublanguage $F \subseteq L_m(P)$ is controllable and forcibly-controllable.
\end{enumerate}
\end{fact}

\begin{proof} ~
\begin{enumerate}
    \item Since $\overline{\varnothing} = \varnothing$, by Definitions~2, 3, and 4, the empty language is vacuously controllable, forcibly-controllable, and forcible.
    \item Since $\overline{L(P)} = L(P)$, we have both $\overline{L(P)} \subseteq L(P)$ and $s\sigma \in L(P) \implies s\sigma\in \overline{L(P)}$ for any $s$ and $\sigma$. Therefor by definition $L(P)$ is controllable, forcibly-controllable, and forcible.
    \item For $\Sigma_f = \varnothing$, the second line of logic formula in Definition~\ref{defn:forcon} is always false. Thus the definition of forcible-controllability reduces to the definition of controllability, and the conclusion holds. 
   \item For $\Sigma_c= \varnothing$ we have $\Sigma  = \Sigma_u$. In this case the definitions of forcible-controllability and forcibility are identical, and hence the conclusion holds.
   \item Let $\Sigma_u = \varnothing$ and $F \subseteq L_m(P)$. 
   According to the logic formulas in the definitions of controllable and forcibly-controllable sublanguages, $F$ is vacuously controllable and forcibly-controllable. 
\end{enumerate}
\end{proof}

Now we state the main result of this section, which characterizes solvability of Problem~\ref{problem:fsc} by forcible-controllability.

\begin{theorem} \label{thm:solvability}
Consider a plant $P$ and a specification $\varnothing \neq F \subseteq L_m(P)$.  There exists a supervisory control $V : L(P) \rightarrow Pwr(\Sigma)$ such that $V/P$ is nonblocking and $L_m(V/P) = F$ 
iff $F$ is forcibly-controllable.
\end{theorem}

The result of Theorem~\ref{thm:solvability}, as well as its proof below, is a direct generalization of the fundamental result of supervisory control theory \cite{RW} by adding the new control mechanism of forcible events. In the proof, for convenience we introduce the notation $E_F(s) = \{\sigma \in \Sigma \mid s\sigma \in \overline{F}\}$, which collects the events that can continue after a string $s$ in $\overline{F}$. If $F \subseteq L_m(P)$, then $E_F(s) \subseteq E_P(s)$. 

\begin{proof}
Define a supervisory control $V : L(P) \rightarrow Pwr(\Sigma)$ according to 
\begin{align} \label{eq:VV}
    V(s) = \begin{cases}
  \Sigma_u \cup \{ \sigma \in \Sigma_c \mid s\sigma \in \overline{F}\}  & \text{ if } \big[ E_F(s) \cap \Sigma_f = \varnothing \\
  & \hspace{-2.1cm}{}  \lor (\forall \sigma \in \Sigma_u)~[ \sigma \in E_P(s) \Rightarrow s\sigma \in \overline{F}] \big] \\
  \{ \sigma \in \Sigma_f \mid s\sigma \in \overline{F}\} & \text{ if } \big[ E_F(s) \cap \Sigma_f \neq \varnothing \\
  & \hspace{-2.0cm} {} \land (\exists \sigma \in \Sigma_u)~[ \sigma \in E_P(s) \land s\sigma \notin \overline{F}] \big].
\end{cases}
\end{align}
Since $E_F(s) \cap \Sigma_f \neq \emptyset$ implies $E_P(s) \cap \Sigma_f \neq \emptyset$, this $V$ is indeed a supervisory control as defined in Equation~(\ref{eq:V}).  This supervisory control $V$ will force a forcible event $f$ after string $s$ only when $f \in E_F(s)$ and there is an uncontrollable event $u$ s.t. $su \in L(P) \setminus \overline{F}$ (i.e. controllability fails). 

(Sufficiency) We claim that with the above $V$ in (\ref{eq:VV}),
\begin{align*}
L(V/P) = \overline{F}.
\end{align*} 
Before proving this claim, we point out that once this claim is established, the desired conclusion follows immediately:
\begin{align*}
\overline{L_m(V/P)} &= \overline{F} = L(V/P) \mbox{~~~($V/P$ is nonblocking)} \\
L_m(V/P) &= L(V/P) \cap F = \overline{F} \cap F = F.
\end{align*} 
Now we prove the claim, by induction on the length of strings. For the base case, since the empty string $\epsilon$ belongs to both $L(V/P)$ (by definition) and $\overline{F}$ (since $F \neq \varnothing$), the conclusion holds. Now let $s\in \Sigma^*$ and suppose that 
\begin{align*}
s \in L(V/P) \Leftrightarrow s \in \overline{F}.
\end{align*} 
Let $\sigma \in \Sigma$. First suppose $s\sigma \in L(V/P)$; it will be shown $s\sigma \in \overline{F}$. By the definition of $L(V/P)$, we have 
\begin{align*}
s \in L(V/P),~ \sigma \in V(s),~ \sigma \in E_P(s). 
\end{align*} 
We also have by hypothesis that $s \in \overline{F}$. Now we consider two cases, according to the supervisory control $V$ in Equation (\ref{eq:VV}). 

\noindent Case 1: $E_F(s) \cap \Sigma_f = \varnothing \lor (\forall \sigma \in \Sigma_u)~[ \sigma \in E_P(s) \Rightarrow s\sigma \in \overline{F}]$.

In this case, if $\sigma \in \Sigma_u$ and $E_F(s) \cap \Sigma_f = \varnothing$, since $F$ is forcibly controllable and $\sigma \in E_P(s)$, we have $s\sigma \in \overline{F}$ (first logic formula in Definition~\ref{defn:forcon}). If $\sigma \in \Sigma_u$ and $(\forall \sigma \in \Sigma_u)~[ \sigma \in E_P(s) \Rightarrow s\sigma \in \overline{F}]$, since $\sigma \in E_P(s)$, we again have $s\sigma \in \overline{F}$. Finally if $\sigma \in \Sigma_c$, since $\sigma \in V(s)$, we have $s\sigma \in \overline{F}$ (by $V(s)$ in Equation  (\ref{eq:VV})). 

\noindent Case 2: $E_F(s) \cap \Sigma_f \neq \varnothing \land (\exists \sigma \in \Sigma_u)~[ \sigma \in E_P(s) \land s\sigma \notin \overline{F}]$.

In this case, since $\sigma \in V(s)$, we have $\sigma \in \Sigma_f$ and $s\sigma \in \overline{F}$.

\noindent Therefore in both cases above, we have established $s\sigma \in \overline{F}$. 

Conversely suppose $s\sigma \in \overline{F}$; we show $s\sigma \in L(V/P)$. It follows from $s\sigma \in \overline{F}$ that
\begin{align*}
& s\sigma \in L(P) \quad \mbox{($\overline{F} \subseteq L(P)$)} \\
& s \in \overline{F} \quad \mbox{($\overline{F}$ is closed)} \\ 
& s \in L(V/P) \quad \mbox{(hypothesis).}
\end{align*}
Hence to show $s\sigma \in L(V/P)$, all we need to show is $\sigma \in V(s)$. Again we consider two cases like above. 

\noindent Case 1: $E_F(s) \cap \Sigma_f = \varnothing \lor (\forall \sigma \in \Sigma_u)~[ \sigma \in E_P(s) \Rightarrow s\sigma \in \overline{F}]$.

In this case, if $\sigma \in \Sigma_u$, then $\sigma \in V(s)$ (by $V(s)$ in Equation (\ref{eq:VV})). If $\sigma \in \Sigma_c$, since $s\sigma \in \overline{F}$, it again follows from $V(s)$ in Equation (\ref{eq:VV}) that $\sigma \in V(s)$. 

\noindent Case 2: $E_F(s) \cap \Sigma_f \neq \varnothing \land (\exists \sigma \in \Sigma_u) ~[ \sigma \in E_P(s) \land s\sigma \notin \overline{F}]$.

In this case, since $s\sigma \in \overline{F}$ and $F$ is forcibly-controllable (in particular the second logic formula in Definition~\ref{defn:forcon}), we have $\sigma \in \Sigma_f$. Hence $\sigma \in V(s)$. 

\noindent Therefore in both cases above, we have established $s\sigma \in V(s)$, and thereby $s\sigma \in L(V/P)$ is proved.  
This finishes the induction step, and the proof of sufficiency is now complete.

(Necessity) Suppose that $\overline{L_m(V/P)} = L(V/P)$ and $L_m(V/P) = F$. It will be shown that $F$ is forcibly-controllable. By the above assumption, we have $L(V/P) = \overline{L_m(V/P)} = \overline{F}$. Let $s \in \overline{F}$. Then $s \in L(V/P)$. Below we consider two cases (as done in the sufficiency proof). 

\noindent Case 1: $E_F(s) \cap \Sigma_f = \varnothing \lor (\forall \sigma \in \Sigma_u)~[ \sigma \in E_P(s) \Rightarrow s\sigma \in \overline{F}]$.

If $E_P(s) \cap \Sigma_f = \varnothing$, let $\sigma \in \Sigma_u$ and $\sigma \in E_P(s)$. Then $\sigma \in V(s)$ and $s\sigma \in L(P)$. It follows from the definition of $L(V/P)$ that $s\sigma \in L(V/P)$. Hence $s\sigma \in \overline{F}$. This shows that $F$ is controllable, and therefore is forcibly-controllable. On the other hand, if $(\forall \sigma \in \Sigma_u) ~[\sigma \in E_P(s) \Rightarrow s\sigma \in \overline{F}]$, by definition $F$ is forcibly-controllable. 

\noindent Case 2: $E_F(s) \cap \Sigma_f \neq \varnothing \land (\exists \sigma \in \Sigma_u)~[ \sigma \in E_P(s) \land s\sigma \notin \overline{F}]$.

In this case, since $E_F(s) \cap \Sigma_f \neq \varnothing$, there exists $f \in E_P(s) \cap \Sigma_f$. This means that $(\exists f \in \Sigma_f)sf \in \overline{F}$. Moreover, let $\sigma \in \Sigma \setminus \Sigma_f$. Then by definition of $V(s)$ in (\ref{eq:VV}), we have $\sigma \notin V(s)$. So $s\sigma \notin L(V/P) = \overline{F}$. This proves $(\forall \sigma \in \Sigma \setminus \Sigma_f) s\sigma \notin \overline{F}$. 
Therefore by Definition~\ref{defn:forcon}, $F$ is forcibly-controllable.
\end{proof}

Theorem~\ref{thm:solvability} provides a necessary and sufficient condition for the existence of a nonblocking supervisory control $V$ that realizes an imposed specification $F$, as required in Problem~\ref{problem:fsc}. This supervisory control $V$ may be implemented by a finite automaton $S$ in the same way as \cite[section 3.6]{bookcai}, such that the synchronous product of plant $P$ and $S$ is nonblocking and represents the specification $F$. This automaton $S$ is called a {\em supervisor}.

\section{Supremal Forcibly-Controllable Sublanguages}
\label{section:supremalFC}

In the preceding section, we know from Theorem~\ref{thm:solvability} that Problem~\ref{problem:fsc} is solvable whenever the specification $F \subseteq L_m(P)$ is forcibly-controllable. In this section, we inquire: {\em what if $F$ is not forcibly-controllable? In this case, does there exist an `optimal' supervisor in a sense of maximal permissiveness?}

Let $F \subseteq L_m(P)$. Whether or not $F$ is forcibly controllable, write the family 
\begin{align} \label{forconfaimily}
\mathcal{F}(F) := \{ K \subseteq F \mid K~\mbox{is forcibly-controllable w.r.t.}~P \}. 
\end{align}
In words, $\mathcal{F}(F)$ is the collection of all subsets of $F$ that are forcibly-controllable. Note that $\mathcal{F}(F)$ is nonempty because the empty language $\emptyset$ belongs to it (see {\bf Fact 2} above). Moreover, the next result asserts that $\mathcal{F}(F)$ is closed under set unions.

\begin{prop} \label{prop:union}
Consider $\mathcal{F}(F)$ in (\ref{forconfaimily}) and an arbitrary subset $\{K_i \mid i \in \mathcal{I}\} \subseteq \mathcal{F}(F)$ ($\mathcal{I}$ is some index set). Then $\bigcup_{i \in \mathcal{I}} K_i \in \mathcal{F}(F)$.
\end{prop}

\begin{proof}
Let $K_i$ ($i \in \mathcal{I}$) be a forcibly-controllable sublanguage of $F$ w.r.t.\ $P$. First $\bigcup_{i \in \mathcal{I}} K_i$ is evidently a sublanguage of $F$. Now let $s \in \overline{\bigcup_{i \in \mathcal{I}} K_i}$. We need to show that $(\forall \sigma \in \Sigma_u)~[s\sigma \in L(P) \implies s\sigma \in \overline{\bigcup_{i \in \mathcal{I}} K_i}] \lor ((\exists f \in \Sigma_f)~[sf \in \overline{\bigcup_{i \in \mathcal{I}} K_i}] \land (\forall \sigma \in \Sigma\setminus\Sigma_f) ~[s\sigma \not\in \overline{\bigcup_{i \in \mathcal{I}} K_i}])$. This is equivalent to show that if 
\begin{equation}
    \neg (\forall \sigma \in \Sigma_u)~[s\sigma \in L(P) \implies s\sigma \in \overline{\bigcup_{i \in \mathcal{I}} K_i}] \label{as2}\end{equation}  then $(\exists {f \in \Sigma_f})~[sf \in \overline{\bigcup_{i \in \mathcal{I}} K_i}] \land (\forall \sigma \in \Sigma\setminus\Sigma_f) ~[s\sigma \not\in \overline{\bigcup_{i \in \mathcal{I}} K_i}]$.

Suppose that (\ref{as2}) holds. Since $s \in \overline{\bigcup_{i \in \mathcal{I}} K_i}$, we have that there exists $j\in \mathcal{I}$ such that $s \in \overline{K_j}$. 
Since $s \in \overline{K_j}$ and $K_j$ is forcibly-controllable, we have 
    \begin{itemize}
        \item $(\forall \sigma \in \Sigma_u)~[s\sigma \in L(P) \implies s\sigma \in \overline{K_j}] $, or
        \item $(\exists {f \in \Sigma_f})~[sf \in \overline{K_j}] \land (\forall \sigma \in \Sigma\setminus\Sigma_f) ~[s\sigma \not\in \overline{K_j}]$. 
    \end{itemize}
The first case is impossible due to (\ref{as2}). Thus we only need to consider the second case. 
 It follows from $(\exists {f \in \Sigma_f})~[sf \in \overline{K_j}]$ that $(\exists {f \in \Sigma_f})~[sf \in \overline{\bigcup_{i \in \mathcal{I}} K_i}]$. Now let $\sigma \in \Sigma \setminus \Sigma_f$. Thus $s\sigma \notin \overline{K_j}$. To show that $s\sigma \notin \overline{\bigcup_{i \in \mathcal{I}} K_i}$, we must prove $ s\sigma \notin \overline{\bigcup_{i \in \mathcal{I}} K_i} \setminus \overline{K_j}$. Suppose on the contrary that $s\sigma \in \overline{\bigcup_{i \in \mathcal{I}} K_i} \setminus \overline{K_j}$; then there exists $l \in \mathcal{I}$ such that $l \neq j$ and $s\sigma \in \overline{K_l}$. Hence $s \in \overline{K_l}$ Since $K_l$ is also forcibly-controllable, we have 
    \begin{itemize}
        \item $(\forall \sigma \in \Sigma_u)~[s\sigma \in L(P) \implies s\sigma \in \overline{K_l}] $, or
        \item $(\exists {f \in \Sigma_f})~[sf \in \overline{K_l}] \land (\forall \sigma \in \Sigma\setminus\Sigma_f) ~[s\sigma \not\in \overline{K_l}]$. 
    \end{itemize}
 The first case is again impossible due to (\ref{as2}). 
 From the second case, we derive $s\sigma \notin \overline{K_l}$, which directly contradicts the assumption that $s\sigma \in \overline{K_l}$. Therefore $s\sigma \notin \overline{\bigcup_{i \in \mathcal{I}} K_i} \setminus \overline{K_j}$ after all, and consequently $s\sigma \notin \overline{\bigcup_{i \in \mathcal{I}} K_i}$. This establishes that $\bigcup_{i \in \mathcal{I}} K_i$ is forcibly-controllable.
\end{proof}

Based on Proposition~\ref{prop:union}, $\mathcal{F}(F)$ in (\ref{forconfaimily}) is closed under set unions, and therefore contains a unique supremal element which is the union of all members in $\mathcal{F}(F)$:
\begin{align}
\sup \mathcal{F}(F) := \bigcup \{K \mid K \in \mathcal{F}(F)\}.
\end{align}
This $\sup \mathcal{F}(F)$ is the largest forcibly-controllable sublanguage of $F \subseteq L_m(P)$. If $F$ was already forcibly-controllable, then $\sup \mathcal{F}(F) = F$. 

\begin{theorem} \label{thm:sup}
Consider a plant $P$ and a specification $\varnothing \neq F \subseteq L_m(P)$.  Let $F_{\sup} := \sup \mathcal{F}(F)$. If $F_{\sup} \neq \emptyset$, then there exists a supervisory control $V_{\sup} : L(P) \rightarrow Pwr(\Sigma)$ such that $V_{\sup}/P$ is nonblocking and $L_m(V_{\sup}/P) = F$.
\end{theorem}

\begin{proof}
Since $F_{\sup}$ is forcibly-controllable and nonempty, the conclusion follows from Theorem~\ref{thm:solvability}.
\end{proof}

Theorem~\ref{thm:sup} generalizes the classical result in \cite{WR} with the forcing mechanism incorporated.  
The supervisory control $V_{\sup}$ in Theorem~\ref{thm:sup} realizing the supremal forcibly-controllable sublanguage of $F$ may also be implemented by a finite automaton, or a supervisor (cf. discussion at the end of Section~III). Owing to `supremum' of the realized language $F_{\sup}$, we call this supervisor the {\em maximally permissive supervisor}.  In the next section, we present an algorithm to synthesize such maximally permissive supervisors.

\begin{remark}
We underline an important point about `maximal permissiveness'. The `maximally permissive supervisor' we just mentioned has maximal permissiveness in two meanings. For the first (obvious) one, the language realized by this supervisor is the supremal forcibly-controllable sublanguage $F_{\sup}$. The second meaning of maximal permissiveness is that this supervisor permits `maximal freedom' of choosing forcible events for preempting. Indeed, whenever preempting is needed, all forcible events that can keep the supremal $F_{\sup}$ invariant are available in this suprevisor to be chosen. Having said the above, we point out that in implementation of forcing control, eventually only one forcible event (possibly among multiple valid choices) is chosen. In other words, one can only force one forcible event at a time, and in this sense `maximal permissiveness' is not possessed by our supervisor. In fact in this case, a maximally permissive supervisor generally does not exist. For theoretic soundness and correspondence with the conventional supervisory control, we adopt the term `maximally permissive supervisor' as that realizes the supremal forcibly-controllable sublanguage $F_{\sup}$.  
\end{remark}

Before ending this section, we collect several facts on forcibly-controllable and forcible sublanguages. Not all of them are needed in the sequel, but they are interesting in their own right.

\begin{fact}
Let plant $P$ be a finite automaton over $\Sigma = \Sigma_c \cup \Sigma_u$ and a forcible event set $\Sigma_f \subseteq \Sigma$. Also let $F \subseteq L_m(P)$. 
\begin{enumerate}
    \item The union of forcibly-controllable sublanguages of $F$ is also a forcibly-controllable sublanguage of $F$. 
    \item The union of forcible sublanguages of $F$ is also a forcible sublanguage of $F$. 
    \item The intersection of forcibly-controllable sublanguages of $F$ is not necessarily a forcibly-controllable sublanguage of $F$.
    \item The intersection of forcible sublanguages of $F$ is not necessarily a forcible sublanguage of $F$.
\end{enumerate}
\end{fact}
   
The first fact above is Proposition~\ref{prop:union}, while the second about forcible sublanguages is a straightforward corollary of Proposition~\ref{prop:union}.
The third and fourth facts are illustrated by
Example~\ref{ex:forcon_inter} below. 
These facts together indicate (in terms of algebraic structures) that the set of all forcibly-controllable (or forcible) sublanguages of a given language $F$ is a {\em complete upper semilattice} with the set union operation of the complete lattice of all sublanguages of $F$.
    
\begin{example} \label{ex:forcon_inter}
Consider plant $P$ with $L(P) = L_m(P) = \{\varepsilon, f_1, f_2, u\}$ and $F = \{\varepsilon, f_1, f_2\}$; here $f_1,f_2$ are forcible and $u$ is uncontrollable. 
Also consider two sublanguages $K_1 = \{ \varepsilon,  f_1 \}$ and $K_2 = \{ \varepsilon, f_2 \}$ of $F$. Now, it can be verified that both $K_1$ and $K_2$ are forcibly-controllable and forcible (without controllable events, forcible-controllability and forcibility are equivalent: {\bf Fact 2}). However, $K_1 \cap K_2 = \{ \varepsilon \}$ is neither forcibly-controllable nor forcible (due to the absence of forcible events and the presence of the uncontrollable event $u$ enabled after $\varepsilon$ in the plant).
\end{example}

\section{Synthesizing the maximally permissive forcibly-controllable nonblocking supervisor}
\label{section:algorithm}

In this section we present an algorithm to compute maximally permissive supervisors introduced in the preceding section (see in particular Theorem~\ref{thm:sup} and the paragraph following Theorem~\ref{thm:sup}). 

So far we have considered the setting where a plant automaton $P$ and a specification language $F \subseteq L_m(P)$ are given. Now assume that $F$ is a regular sublanguage, so $F$ can be represented by a finite automaton. 
It is well-known that the conventional synthesis problem with a regular specification language is easily transformed into a synthesis problem where only a plant automaton is considered by applying a so-called {\em plantification transformation} on the finite automaton representing the specification \cite{flordal2007compositional}. This leads to the simplified (but equivalent) problem of supervisory control synthesis below.

\begin{problem}[Synthesis of maximally permissive forcibly-controllable nonblocking supervisor]
Given a plant automaton $P=(Q,\Sigma,\delta,q_0,Q_m)$ with sets $\Sigma_u \subseteq \Sigma$ of uncontrollable events and $\Sigma_f \subseteq \Sigma$ of forcible events, synthesize a maximally permissive, forcibly-controllable, nonblocking supervisor $S$ for $P$.
\end{problem}

In Algorithm \ref{alg:scs_fc} below, a maximally permissive, forcibly-controllabe, nonblocking supervisor is computed starting from the plant automaton (as is common in SCT). Its structure is based on the well-known algorithm for synthesis of maximally permissive supervisors for EFA (without using variables though) \cite{ouedraogo2011nonblocking}. Besides sets of nonblocking and bad states, in this case also sets of states $F_k$ and $F_k^j$ where forcing is applied are maintained.

\renewcommand{\algorithmicrequire}{\textbf{Input:}}
\renewcommand{\algorithmicensure}{\textbf{Output:}}

\begin{algorithm*}
\caption{Supervisory controller synthesis for maximally permissive, forcibly-controllable, nonblocking supervisor}\label{alg:scs_fc}
\begin{algorithmic}[1]
\Require plant $P = (Q,\Sigma,\delta,q_0,Q_m)$, uncontrollable events $\Sigma_u$, \textcolor{red}{forcible events $\Sigma_f$}
\Ensure $S$ is maximally permissive, \textcolor{red}{forcibly-}controllable, nonblocking supervisor for $P$ if it exists, and empty supervisor otherwise

\State $Q_0 \gets Q$
\State \textcolor{red}{$F_0 \gets \varnothing$} \label{lineF0}
\State \textcolor{red}{$\delta_0 \gets \delta$} \qquad  \hfill \textcolor{green}{$S_0 \gets (Q_0,\Sigma,\delta_0,q_0,Q_0 \cap Q_m)$}
\State $k\gets 0$ 
\Repeat \label{automatoniterationstart}
\State $NB_{k+1}^0 \gets Q_m \cap Q_k$
\State $l \gets 0$ 
\Repeat \label{startnonblocking}
\State $\begin{array}[t]{@{}lcl} NB_{k+1}^{l+1} &\gets& NB_{k+1}^l \\
&\cup& \{ q \in Q_k \mid (\exists {\sigma \in \Sigma})~[\delta_{\textcolor{red}{k}}(q,\sigma) \in NB_{k+1}^l] \}\end{array}$
\State $l \gets l+1$
\Until $NB_{k+1}^{l} = NB_{k+1}^{l-1}$ \label{endnonblocking}
\State $NB_{k+1} \gets NB_{k+1}^l$
\State $B_{k+1}^0 \gets Q_k \setminus \NB_{k+1}$ \label{lineB0k+1}
\State $\textcolor{red}{F_{k+1}^0 \gets F_{k} \cap Q_k}$
\label{lineF0k+1}
\State $j \gets 0$
\Repeat \label{startbadstate}
\State \label{lineBj+1k+1} $\begin{array}[t]{@{}lcl} B_{k+1}^{j+1} &\gets& B_{k+1}^j \cup  \left\{ q \in Q_k \left|  (\exists {u \in \Sigma_u})~[\delta_k(q,u) \in B_{k+1}^j] 
\textcolor{red}{{} \land (\forall f \in \Sigma_f)~[\delta_k(q,f) \in B_{k+1}^{j}]} \right. \right\} \\
\end{array}$
\State \label{lineFj+1k+1} \textcolor{red}{$\begin{array}[t]{@{}lcl}F_{k+1}^{j+1} &\gets& F_{k+1}^j 
\cup {} \left\{ q \in Q_k \left|  (\exists {u \in \Sigma_u})~[\delta_k(q,u) \in B_{k+1}^{j}] \land {} 
(\exists {f \in \Sigma_f})~[\delta_k(q,f) \not\in B_{k+1}^{j}]  \right. \right\} \end{array}$}
\State $j \gets j+1$
\Until $B_{k+1}^j = B_{k+1}^{j-1} \textcolor{red}{~\land ~ (\forall q \in F_{k+1}^{j} \setminus B_{k+1}^{j})~[(\exists {f \in \Sigma_f})~[\delta_k(q,f) \not\in B_{k+1}^{j}]]}$ \label{endbadstate}
\State $B_{k+1} \gets B_{k+1}^j$ \label{lineBk+1}
\State \textcolor{red}{$F_{k+1} \gets F_{k+1}^{j}$} \label{lineFk+1}
\State $Q_{k+1} \gets Q_k \setminus B_{k+1}$ \label{lineQk+1}
\State \textcolor{red}{$\delta_{k+1} \gets (\delta_k \cap (Q_{k+1} \times \Sigma \times Q_{k+1})) \setminus (F_{k+1} \times \Sigma\setminus\Sigma_f \times Q_{k+1})$} \label{linedeltak+1} \hfill \textcolor{green}{$S_{k+1} \gets (Q_{k+1},\Sigma,\delta_{k+1},q_0,Q_m \cap Q_{k+1})$}
\State $k \gets k+1$
\Until $Q_k = Q_{k-1} \textcolor{red}{{} \land \delta_{k} = \delta_{k-1}}$ \label{automatoniterationend}
\end{algorithmic}
\end{algorithm*}

Controllable sublanguages have the transitivity property, i.e., whenever $F_1$ is a controllable sublanguage of $F_2$ and $F_2$ is a controllable sublanguage of $F_3$, then also $F_1$ is a controllable sublanguage of $F_3$.
A similar useful property does not hold for forcibly-controllable sublanguages. Consider the languages $F_1 = \{ \varepsilon \}$, $F_2 = \{ f \}$, and $F_3 = \{ f,u \}$ where $f \in \Sigma_f$ and $u \in \Sigma_u$. It is easily verified that $F_1$ is a forcibly-controllable sublanguage of $F_2$ and $F_2$ is a forcibly-controllable sublanguage of $F_3$. However $F_1$ is not a forcibly-controllable sublanguage of $F_3$.

The reason for the absence of this property is the situation that a forcible event ($f$) that has been used to preempt an uncontrollable event ($u$) is later disabled (which means that the uncontrollable event in hindsight cannot be preempted and thus has to be maintained. 

This difference between traditional SCT and SCT with forcible events results in the need to maintain a set of states from which forcing (by a forcible event) has been used to preempt an uncontrollable transition to a bad state. As soon as all forcible events from such a forcing state need to be disabled (because they all lead to a bad state) the forcing state itself becomes bad (Line \ref{lineBj+1k+1}).

\begin{theorem}[Termination]
Algorithm \ref{alg:scs_fc} terminates.
\end{theorem}

\begin{proof} 
First, consider the iteration represented by the repeat-until in Lines \ref{startnonblocking} - \ref{endnonblocking}. In each execution of the body of this repeat-until either a state is added to a set $\NB_{k+1}^{l+1}$ (compared to the set of nonblocking states computed in the previous execution of that body $\NB_{k+1}^l$), or the set $\NB_{k+1}^{l+1}$ is the same as the set $\NB_{k+1}^{l}$. The latter leads to termination, and the former can only take place finitely often since $\NB_{k+1}^{l} \subseteq Q_k \subseteq Q$ for all $k$ and $l$ and the set of states $Q$ is finite.

Second, consider the iteration represented by the repeat-until in Lines \ref{startbadstate} - \ref{endbadstate}. Again, as in the previous case, each execution of the body of the repeat-until results in adding at least one state to the set $B_{k+1}^{j+1}$ (compared to the `previous' set $B_{k+1}^j$) or results in the same set of bad states. The former situation can only take place finitely often as $B_{k+1}^j \subseteq Q_k \subseteq Q$ for all $k$ and $j$. The latter situation leads to termination in case the other condition of the guard in Line \ref{endbadstate} is satisfied. If that condition is not satisfied, then $(\exists {q \in F_{k+1}^j \setminus B_{k+1}^j})~[(\forall f \in \Sigma_f)~[\delta_k(q,f) \in B^{j}_{k+1}]]$. But then, in the next execution of the body of the repeat-until, at least one state is added to the set of bad states of that iteration. 

Finally, consider the outermost repeat-until construction (Lines \ref{automatoniterationstart} - \ref{automatoniterationend}). Invariantly it is the case that $Q_{k+1} \subseteq Q_{k}$ and $\delta_{k+1} \subseteq \delta_k$ for all $k$. Now, given the computations of $Q_{k+1}$ in Line \ref{lineQk+1} and $\delta_{k+1}$ in Line \ref{linedeltak+1}, there are three cases:
\begin{itemize}
    \item $Q_{k+1} \subset Q_k$
    \item $\delta_{k+1} \subset \delta_k$
    \item $Q_{k+1} = Q_k$ and $\delta_{k+1} = \delta_k$.
\end{itemize}
Therefore, either the repeat-until terminates because the first two cases can only be repeated finitely often (there is only a finite number of states and a finite number of transitions in the input plant), or the termination condition of the repeat-until becomes true and the repeat-until terminates.
\end{proof}

\begin{theorem}[Correctness]
Algorithm \ref{alg:scs_fc} computes a maximally permissive, forcibly-controllable, nonblocking supervisor.
\end{theorem}

\begin{proof} 
For presenting the proof (of forcible-controllability) we introduce thought variables $S_k = (Q_k,\Sigma,\delta_k,q_0,Q_k \cap Q_m)$ which represent the subautomata of $P$ derived after iteration $k$ of the outer iteration. See the green lines of code and note that the variables $S_k$ are never used to update any other variable, nor to evaluate a condition in the control flow of the pseudo code. Consequently, these lines have no impact on the outcome of the algorithm and are only used for facilitating the proof.

\paragraph*{Nonblocking}
Upon termination of Algorithm \ref{alg:scs_fc}, say after $K$ iterations, we have $Q_K = Q_{K-1}$ and $\delta_K = \delta_{K-1}$. This can only be the case if $B_K= \varnothing$ and therefore only if $NB_K = Q_K$. So, if we prove that $NB_{K} = \{ q \in Q_{K-1} \mid q \mbox{ is nonblocking }\}$ then we are done, since $Q_{K-1} = Q_K$. The pseudo code computing the sets of nonblocking states is fairly standard and it has been shown before that indeed the set of nonblocking states is computed. See for example \cite{ouedraogo2011nonblocking} or \cite{Thuijsman2020} for a statement to this effect.

\paragraph*{Forcible-controllability}
We will prove that the output of the algorithm is forcibly-controllable w.r.t.\ the input plant $P$. First, we prove the properties presented in Equations \ref{forcinginvariant} and \ref{invariantcontrollability} below.
\begin{equation}\label{forcinginvariant}
\begin{array}{l}
(\forall 0 \leq k \leq K,~\forall q \in F_k\setminus B_k) \\ 
\hspace*{2cm} \begin{array}{l}
[(\exists {f \in \Sigma_f})~[\delta_k(q,f) \in Q_k] \\
\land \\ (\forall n \in \Sigma\setminus\Sigma_f)~[\neg\delta_k(q,n)! ]]. 
\end{array}
\end{array}
\end{equation}

The proof of the property in Equation \ref{forcinginvariant} is by induction on $k$. For the base case $k=0$ this follows immediately from the fact that $F_0 = \varnothing$ (Line \ref{lineF0}). For the induction case, assume that $(\forall q \in F_k\setminus B_k)~[(\exists {f \in \Sigma_f})~[\delta_k(q,f) \in Q_k] \land (\forall n \in \Sigma\setminus\Sigma_f)~[\neg\delta_k(q,n)! ]]$. Let $q \in F_{k+1}\setminus B_{k+1}$. Then, $q \in F_{k+1}^{J}$ and $q \not\in B_{k+1}^J$ for the $J$ that corresponds with  the value of $j$ upon termination of this iteration of the repeat-until statement in Lines \ref{startbadstate} - \ref{endbadstate}. Then, also $(\forall q \in F_{k+1}^J\setminus B_{k+1}^J)~[(\exists {f \in \Sigma_f})~[\delta_k(q,f) \not\in B_{k+1}^J]]$ as this is part of the termination condition of the repeat-until. Then, as $q \in F_{k+1}^J \setminus B_{k+1}^J$, we have $(\exists {f \in \Sigma_f})~[\delta_k(q,f) \not\in B_{k+1}^J]$. Then also, following Line \ref{lineQk+1}, $(\exists {f \in \Sigma_f})~[\delta_k(q,f) \in Q_{k+1}]$. Let $n \in \Sigma\setminus\Sigma_f$. As $q \in F_{k+1}$, from Line \ref{linedeltak+1}, it follows that $\neg \delta_{k+1}(q,n)!$.
Consequently, Equation \ref{forcinginvariant} holds.

\begin{equation}
\begin{array}{l}
(\forall 0 \leq k \leq K, ~\forall q \in Q_k \setminus F_k,~\forall u \in \Sigma_u) \\
\hspace*{2cm}
[\delta_0(q,u) \in Q_0 \implies \delta_k(q,u) \in Q_k] 
\end{array}
    \label{invariantcontrollability}
\end{equation}

The proof of this property is by induction on $k$. The base case $k=0$ follows trivially. For the induction step, assume that $(\forall q \in Q_k \setminus F_k, ~\forall {u \in \Sigma_u})~[\delta_0(q,u) \in Q_0 \implies \delta_k(q,u) \in Q_k]$. Let $q \in Q_{k+1}\setminus F_{k+1}$ and $u \in \Sigma_u$. Assume that $\delta_0(q,u) \in Q_0$. As $q \not\in F_{k+1}$, it follows that $\delta_{k+1}(q,u) = \delta_k(q,u)$ provided that $\delta_{k+1}(q,u) \in Q_{k+1}$. Towards a contradiction assume that $\delta_k(q,u) \in Q_k\setminus Q_{k+1}$. Then $\delta_k(q,u) \in B_{k+1}$. Then $\delta_k(q,u) \in B_{k+1}^J$ and because of the termination condition of the repeat-until, also $\delta_k(q,u) \in B_{k+1}^{J-1}$. But then, either $q \in B_{k+1}^J$ or $q \in F_{k+1}^J$. This contradicts the assumption that $q \in Q_{k+1}\setminus F_{k+1}$.

Together the properties in Equation \ref{forcinginvariant} and Equation \ref{invariantcontrollability}, show that each $S_k$ as produced during the execution of the algorithm is forcibly-controllable w.r.t.\ $P$. Each of the states of $S_k$, i.e., the set $Q_k$, either respects controllability (Equation \ref{invariantcontrollability} holds for the states from $Q_k\setminus F_k$) or is forcing and respects Equation \ref{forcinginvariant} (for states from $F_k$ that are in $Q_k$).

\paragraph*{Maximally permissive}
Suppose there is a larger nonblocking, forcibly-controllable subautomaton $S'$ for $P$. Then, it contains a state that $S$ does not, or, if the former does not hold (and both $S'$ and $S$ have the same set of states), it contains a transition between states they both have that $S$ does not have.
Maximal permissiveness is achieved by construction. In the algorithm, states or transitions are only removed if there is necessity for that.

Let us consider the former case. $S'$ contains a state that is not present in $S$. Then this state is removed (in the computation of $S$) in some invocation of Line \ref{lineQk+1}. Thus this state is in $B_{k+1}$, i.e., a bad state. We need to show that any addition of a state to the set of bad states cannot be avoided since otherwise the resulting supervisor violates nonblockingness or forcible-controllability. Proof of this follows the same reasoning as in the underlying algorithm used for synthesis of EFA \cite{ouedraogo2011nonblocking}.

Now, let us consider the latter case. A transition is present in $S'$ but not in $S$. This has to be due to Line \ref{linedeltak+1}. Therefor this concerns a transition from a forcing state labelled by an nonforcible event. Since the source state is forcing (i.e. if it does not exercise its forcing action it is deemed to be bad) we know that this transition would violate forcible-controllability. 
\end{proof}

\begin{theorem}[Complexity]
Algorithm \ref{alg:scs_fc} has worst-case time complexity ${\cal O}(|Q|^2 ~| \Sigma |)$.
\end{theorem}

\begin{proof}
The black part of the algorithm is basically the standard algorithm for supervisory control synthesis without event forcing and it has been established to have a worst-case time complexity of ${\cal O}(|Q|^2 ~| \Sigma |)$  in \cite{ramadge1989control}. Obviously, the red additions (for the purpose of forcing supervisory control) do not negatively, nor positively, affect this worst-case time complexity.
\end{proof}

\begin{remark}
When Algorithm \ref{alg:scs_fc} is applied with $\Sigma_f = \varnothing$, it in fact computes the maximally permissive, controllable, nonblocking supervisor. All sets of states where forcing is applied are empty and hence the red additions to the bad state computation do not add anything. Moreover, there is no need to change the transition relation (and the subscript can be removed).
\end{remark}

\begin{remark}
Replacing all occurrences of the set $\Sigma_u$ in Algorithm \ref{alg:scs_fc} (in Lines \ref{lineBj+1k+1} and \ref{lineFj+1k+1}) by $\Sigma$ results in an algorithm for the maximally permissive, forcible, nonblocking supervisor. 
\end{remark}

\section{Case study - small manufacturing line}
\label{section:smallmanufacturingsystem}

We consider the example of the small manufacturing line that was also used in \cite[section 3.8]{bookcai}. There are two machines ($M_1$ and $M_2$) and one (plantified) specification (namely that finishing processing on machine $M_1$ and starting processing on machine $M_2$ alternate). See Fig.~\ref{plant} for the graphically represented automata. The specification has been formulated in such a way that it is controllable (and hence forcibly-controllable as well).  For now, we consider only the event $start\_M_2$ representing the start of processing on machine $M_2$ as forcible.

\begin{figure}[htbp]
\begin{center}
\begin{tikzpicture}[auto,->,node distance=2.5cm,align=center,font=\smaller\it, uncontrollable/.style={densely dashed}]
        \node[state,minimum size=1cm] (S1) {\it Busy};
        \node[state,initial above,accepting,initial text = {$M_1$}, minimum size=1cm,left of=S1] (S0) {\it Idle};

        \path[->,shorten >=1pt]
            (S0) edge[bend2,above] node{$start\_M_{1}$} (S1)
            (S1) edge[bend2,below,uncontrollable] node{$end\_M_{1}$} (S0)
            ;
    \end{tikzpicture}
\qquad
\begin{tikzpicture}[auto,->,node distance=2.5cm,align=center,font=\smaller\it, uncontrollable/.style={densely dashed}]
        \node[state,minimum size=1cm] (S1) {\it Busy};
        \node[state,initial above,accepting,initial text = {$M_2$}, minimum size=1cm,left of=S1] (S0) {\it Idle};

        \path[->,shorten >=1pt]
            (S0) edge[bend2,above] node{$\underline{start\_M_{2}}$} (S1)
            (S1) edge[bend2,below,uncontrollable] node{$end\_M_{2}$} (S0)
            ;
    \end{tikzpicture}

\begin{tikzpicture}[auto,->,node distance=2.5cm,align=center,font=\smaller\it, uncontrollable/.style={densely dashed}]
        \node[state,initial left,accepting,initial text = {$R$}, minimum size=1cm] (S0) {\it 0};
        \node[state,right of=S0, minimum size=1cm] (S1) {\it 1};
        \node[state,right of=S1, minimum size=1cm] (S2) {\it 2};

        \path[->,shorten >=1pt]
            (S0) edge[bend2,above,uncontrollable] node{$end\_M_{1}$} (S1)
            (S1) edge[bend2,below] node{$\underline{start\_M_{2}}$} (S0)
            (S1) edge[uncontrollable] node{$end\_M_{1}$} (S2)
            (S0) edge  [loop above] node {$start\_M_{1}$} ()
            (S1) edge  [loop above] node {$start\_M_{1}$} ()
            (S0) edge  [uncontrollable,loop below] node {$end\_M_{2}$} ()
            (S1) edge  [uncontrollable,loop below] node {$end\_M_{2}$} ()
            ;
    \end{tikzpicture}
\caption{Automata for the machines and the specification in the small manufacturing line.}\label{plant}
\end{center}
\end{figure}
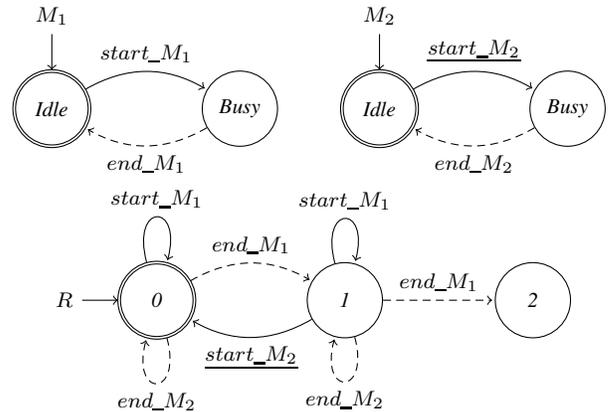

\begin{remark}
The self-loops in the specification automaton are not necessary. We have added them in this case, because as a consequence the number of states in the synchronous product is smaller and that facilitates graphical representation of that synchronous product.
\end{remark}

Fig.~\ref{supervisors} presents a number of relevant automata:
\begin{itemize}
    \item the synchronous product of the three automata from Fig.~\ref{plant} by considering all states and transitions regardless their color. This automaton is the input for the synthesis procedure.
    \item the traditional supervisor (no forcing at all), is obtained by considering only the black colored states and transitions. The states $BI1$ and $BB1$ have been made unreachable by disabling the incoming $start\_M_1$ events. 
    \item the forcing supervisor in the case that only event $start\_M_2$ is forcible is obtained by considering the black and blue states and transitions. The red and cyan states and transitions are omitted by synthesis. Note that this is the same result as predicted in \cite[Section 3.8]{bookcai}.
    \item the forcing supervisor in the case that also event $end\_M_2$ is forcible is given by all but the red states and transitions.
\end{itemize}

Both the synchronous product and the traditional supervisor have been obtained by applying the Compositional Interchange Format (CIF) \cite{van2014cif,ESCET2023}. CIF is part of the Eclipse Supervisory Control Engineering Toolkit (ESCET\texttrademark) project.\footnote{See http://eclipse.org/escet. `Eclipse', `Eclipse ESCET' and `ESCET' are trademarks of Eclipse Foundation, Inc.} Both forcing supervisors have been computed manually following Algorithm \ref{alg:scs_fc}.

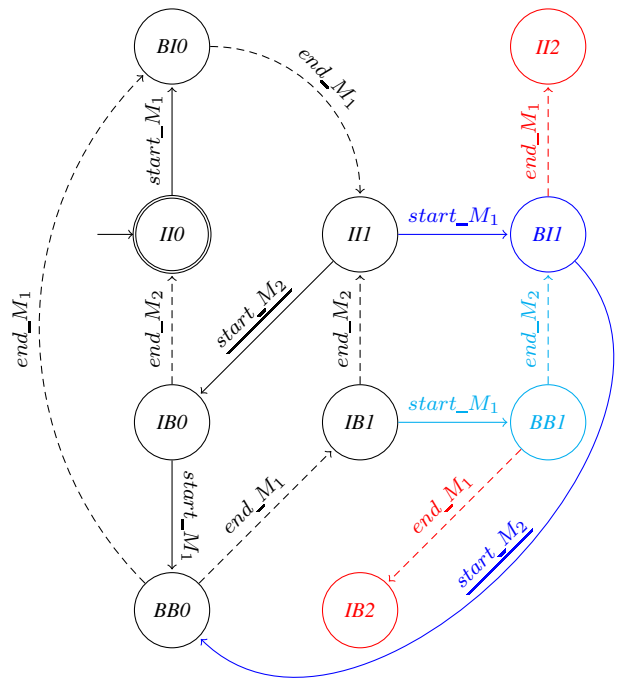
\begin{figure}[htbp]
\begin{tikzpicture}[auto,->,node distance=2.5cm,align=center,font=\smaller\it, uncontrollable/.style={densely dashed}]
        \node[state,initial left,accepting,initial text = {}, minimum size=1cm] (II0) {\it II0};
        \node[state,above of=II0, minimum size=1cm] (BI0) {\it BI0};
        \node[state,below of=II0, minimum size=1cm] (IB0) {\it IB0};
        \node[state,right of=II0, minimum size=1cm] (II1) {\it II1};
        \node[state,below of=IB0, minimum size=1cm] (BB0) {\it BB0};
        \node[state,right of=IB0, minimum size=1cm] (IB1) {\it IB1};
        \node[state,right of=II1, minimum size=1cm,blue] (BI1)
        {\it BI1};
        \node[state,below of=BI1, minimum size=1cm,cyan] (BB1)
        {\it BB1};
        \node[state,below of=IB1, minimum size=1cm,red] (IB2)
        {\it IB2};
        
        \node[state,above of=BI1, minimum size=1cm,red] (II2)
        {\it II2};
        
        \path[->,shorten >=1pt]
            (II0) edge[sloped] node{$start\_M_{1}$} (BI0)
            (II1) edge[sloped,blue] node{$start\_M_{1}$} (BI1)
            (IB0) edge[sloped] node{$start\_M_{1}$} (BB0)
            (IB1) edge[sloped,cyan] node{$start\_M_{1}$} (BB1)
            (II1) edge[sloped] node{$\underline{start\_M_{2}}$} (IB0)
            (BI1) edge[sloped,out=-45,in=-45,distance=3cm,blue] node{$\underline{start\_M_{2}}$} (BB0) 
            (BI0) edge[sloped,uncontrollable,out=0,in=90] node{$end\_M_{1}$} (II1)
            (BB0) edge[sloped,uncontrollable,sloped] node{$end\_M_{1}$} (IB1)
            (BI1) edge[sloped,uncontrollable,red] node{$end\_M_{1}$} (II2)
            (BB1) edge[sloped,uncontrollable,red] node{$end\_M_{1}$} (IB2)
            (IB0) edge[sloped,uncontrollable] node{$end\_M_{2}$} (II0)
            (IB1) edge[sloped,uncontrollable] node{$end\_M_{2}$} (II1)
            (BB1) edge[sloped,uncontrollable,cyan] node{$end\_M_{2}$} (BI1)
            (BB0) edge[sloped,out=135,in=225,uncontrollable] node{$end\_M_{1}$} (BI0) 
            ;
    \end{tikzpicture}
    \caption{Automaton representing the synchronous product and, using colors as described in the text, three different supervisors.}\label{supervisors}
\end{figure}

This example suggests, and our intuitions support this observation, that using additional forcible events generally results in more permissive supervisors.

\section{Case study - small factory}
\label{section:smallfactory}

In this section, both the traditional supervisor and forcing supervisor (for all controllable events) are provided for the small factory as discussed in \cite{bookcai}. The plant consists of two machines, namely $\mathit{M_1}$ and $\mathit{M_2}$ (see Fig.~\ref{fig:plantsmallfactory} for the automata representing the plant components). These machines can start and end processing (by means of controllable events $\mathit{start}\_M_{i}$, for $i=1,2$), but can also break down (uncontrollably with events $\mathit{break}\_M_{i}$, for $i=1,2$) when in the working state ($W$). In case of a breakdown, a repair can be achieved (controllable events $\mathit{repair}\_M_{i}$, for $i = 1,2$) upon which the system starts again in the idle state ($I$). The supervisors developed are supposed to achieve the requirements given by means of the controllable automata specifications\footnote{A controllable automaton is one whose marked language is controllable w.r.t.\ the plant.} given in Fig.~\ref{fig:requirementssmallfactory}\footnote{The models for the plant components and requirements are based on those provided in \cite{bookcai} with easier to interpret names for events and where the automata specifications are made controllable (where needed).}. The first requirement captures that the two machines are placed in line with a single-space buffer in between. Note that this buffer is not modeled explicitly. The second requirement states that in case both machines break down, priority has to be given to repairing machine $M_2$.

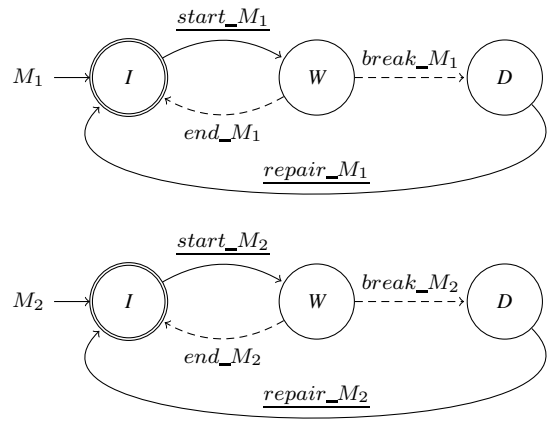
\begin{figure}
    \centering
    \begin{tikzpicture}[auto,->,node distance=2.5cm,align=center,font=\smaller\it, uncontrollable/.style={densely dashed}]
        \node[state,initial left,accepting,initial text = {$M_1$}, minimum size=1cm] (I) {\it I};
        \node[state,right of=I, minimum size=1cm] (W) {\it W};
        \node[state,right of=W, minimum size=1cm] (D) {\it D};

        \path[->,shorten >=1pt]
            (I) edge[bend2,above] node{\underline{$start\_M_{1}$}} (W)
            (W) edge[bend2,below,uncontrollable] node{$end\_M_{1}$} (I)
            (W) edge[uncontrollable] node{$break\_M_{1}$} (D)
            (D) edge  [out=-45,in=225,above] node {\underline{$repair\_M_{1}$}} (I)
            ;
    \end{tikzpicture}
        \begin{tikzpicture}[auto,->,node distance=2.5cm,align=center,font=\smaller\it, uncontrollable/.style={densely dashed}]
        \node[state,initial left,accepting,initial text = {$M_2$}, minimum size=1cm] (I) {\it I};
        \node[state,right of=I, minimum size=1cm] (W) {\it W};
        \node[state,right of=W, minimum size=1cm] (D) {\it D};

        \path[->,shorten >=1pt]
            (I) edge[bend2,above] node{\underline{$start\_M_{2}$}} (W)
            (W) edge[bend2,below,uncontrollable] node{$end\_M_{2}$} (I)
            (W) edge[uncontrollable] node{$break\_M_{2}$} (D)
            (D) edge  [out=-45,in=225,above] node {\underline{$repair\_M_{2}$}} (I)
            ;
    \end{tikzpicture}

    \caption{Automata representing the machines $M_1$ and $M_2$ in the small factory.}
    \label{fig:plantsmallfactory}
\end{figure}

\begin{figure}
    \centering
        \begin{tikzpicture}[auto,->,node distance=2.5cm,align=center,font=\smaller\it, uncontrollable/.style={densely dashed}]
        \node[state,initial left,accepting,initial text = {$R_1$}] (I) {};
        \node[state,right of=I] (W) {};
        \node[state,right of=W] (D) {};

        \path[->,shorten >=1pt]
            (I) edge[uncontrollable,bend2,above] node{$end\_M_{1}$} (W)
            (W) edge[bend2,below] node{\underline{$start\_M_{2}$}} (I)
            (W) edge[uncontrollable] node{$end\_M_{1}$} (D)
            ;
    \end{tikzpicture}
 \begin{tikzpicture}[auto,->,node distance=2.5cm,align=center,font=\smaller\it, uncontrollable/.style={densely dashed}]
        \node[state,initial left,accepting,initial text = {$R_2$}] (N) {};
        \node[state,right of=I] (P) {};

        \path[->,shorten >=1pt]
            (N) edge[loop below] node{\underline{$repair\_M_{1}$}} ()
            (N) edge[bend2,above,uncontrollable] node{$break\_M_{2}$} (P)
            (P) edge[bend2,below] node{\underline{$repair\_M_{2}$}} (N)
            ;
    \end{tikzpicture}
    \caption{Controllable specification automata for the small factory.}
    \label{fig:requirementssmallfactory}
\end{figure}
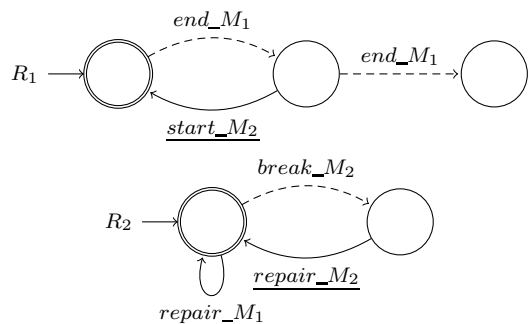

Application of supervisory control synthesis Algorithm~1 to this system (consisting of plant components and controllable requirements) where all (and only) the controllable events are considered forcible) results in a supervisor that can be represented by the automaton in Fig.~\ref{fig:forcingsupervisorsmallfactory}. 

\begin{figure*}[htbp]
    \centering
    \begin{tikzpicture}[auto,->,node distance=4cm,align=center,font=\smaller\it, uncontrollable/.style={densely dashed}]
        \node[state,initial below,accepting,initial text = {}] (s1) {};

        \node[state,above of=s1] (s5) {};  
        \node[state,above of=s5] (s4) {};
        \node[state,left of=s4,fill=green] (s6) {};
        
        \node[state,right of=s5] (s10) {};
        \node[state,left of=s5] (s12) {};

        \node[state,left of=s12] (s16) {}; 
        \node[state,left of=s16] (s3) {};
        \node[state,above of=s4] (s17) {};
        \node[state,right of=s17] (s14) {};
        
        \node[state,left of=s17,fill=green] (s19) {};
        \node[state,above of=s19] (s2) {};
        \node[state,left of=s19] (s7) {};
        \node[state,above of=s2] (s13) {};

        \path[->,shorten >=1pt]
            (s1) edge[bend2,below] node[sloped]{\hspace*{2cm}\underline{$start\_M_{1}$}} (s2)
            (s2) edge[out=180,in=90,above,uncontrollable] node[sloped]{$break\_M_{1}$} (s3)
            (s2) edge[out=-67,in=135,uncontrollable] node[sloped]{\hspace*{2cm}$end\_M_{1}$} (s4)
            (s3) edge[out=270,in=180,below] node{\underline{$repair\_M_{1}$}} (s1)
            (s4) edge[] node[sloped]{\underline{$start\_M_{2}$}} (s5)
            (s4) edge[color=red] node{\underline{$start\_M_{1}$}} (s6)
            (s5) edge[above,uncontrollable] node{$break\_M_{2}$} (s10)
            (s5) edge[uncontrollable] node[sloped]{$end\_M_{2}$} (s1)
            (s5) edge[out=135,in=-75,below] node[sloped]{\underline{$start\_M_{1}$}\hspace*{2cm}} (s7)
            (s6) edge[color=red] node[sloped]{\underline{$start\_M_{2}$}} (s7)
            (s7) edge[bend2,above,uncontrollable] node[sloped]{$break\_M_{2}$} (s13)
            (s7) edge[uncontrollable] node[sloped]{$end\_M_{2}$} (s2)
            (s7) edge[out=245,in=135,above,uncontrollable] node[sloped]{$break\_M_{1}$} (s12)
            (s7) edge[bend2,uncontrollable] node{\hspace*{2cm}$end\_M_{1}$} (s14)
            (s10) edge[below] node[sloped]{\underline{$repair\_M_{2}$}} (s1)
            (s10) edge[out=67,in=0] node[sloped]{\underline{$start\_M_{1}$}} (s13)
            (s12) edge[above,uncontrollable] node{$break\_M_{2}$} (s16)
            (s12) edge[bend2,above,uncontrollable] node{$end\_M_{2}$} (s3)
            (s12) edge[below] node{\underline{$repair\_M_{1}$}} (s5)
            (s13) edge[out=180,in=100,above,uncontrollable] node[sloped]{\hspace*{2cm}$break\_M_{1}$} (s16)
            (s13) edge[out=-45,in=90,uncontrollable] node[sloped]{$end\_M_{1}$} (s17)
            (s13) edge[] node[sloped]{\underline{$repair\_M_{2}$}} (s2)
            (s14) edge[below,uncontrollable] node{$break\_M_{2}$} (s17)
            (s14) edge[below,uncontrollable] node[sloped]{$end\_M_{2}$} (s4)
            (s16) edge[] node{\underline{$repair\_M_{2}$}} (s3)
            (s17) edge[] node[sloped]{\underline{$repair\_M_{2}$}} (s4)
            (s17) edge[above,color=red] node{\hspace*{1cm}\underline{$start\_M_{1}$}} (s19)
            (s19) edge[above,color=red] node[sloped]{\underline{$repair\_M_{2}$}} (s6)
            ;
    \end{tikzpicture}
    \caption{Forcing supervisor for the small factory.}
    \label{fig:forcingsupervisorsmallfactory}
\end{figure*}
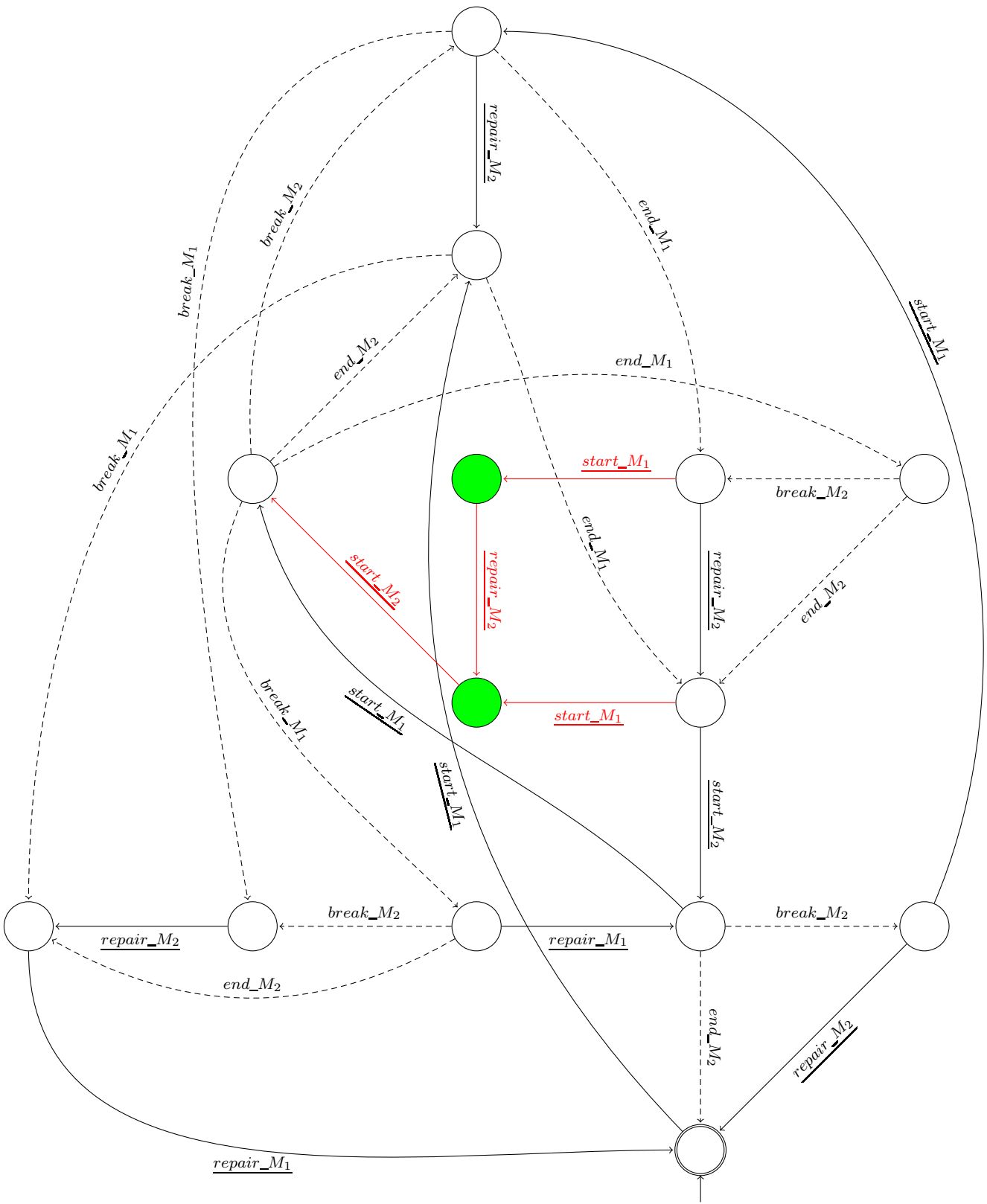

In the figure, the green states are those where forcing of an event is used to prevent reaching a state that is unsafe via an uncontrollable event. In the conventional supervisor those states are made unreachable by preventing their incoming controllable events (the red transitions labelled by the event $\mathit{start}\_M_{1}$). So in the nonforcing supervisor, the green state and the red transitions are absent.
Thus the forcing supervisor indeed allows behaviors that are prohibited by the nonforcing supervisor.

As can be seen from the figure, the forcing is achieved by using the forcible evens $\mathit{start}\_M_2$ and $\mathit{repair}\_M_2$. Therefore, in case the controllable events involving machine $M_1$ are not forcible, the same forcing supervisor results.

\section{Conclusions}
\label{section:conclusions}

In this paper, the setting of supervisory control theory has been enriched with the possibility for a supervisor to force events in order to preempt uncontrollable events from leading the plant into undesired states. This has required an adaptation of the notion of controllability to forcible-controllability. This notion and its properties have been studied in detail and a supervisory control problem for obtaining a maximally permissive, foricbly-controllable, nonblocking supervisory controller has been posed and solved. The approach has been illustrated with two case studies from literature.

Our future work is twofold. In one direction, we aim to extend this framework with forcing to tackle other fundamental problems in supervisory control; these include partial observation \cite{LinWon88}, timed systems \cite{brandin1994supervisory}, and supervisor localization \cite{bookcai2}. In the other direction, with the forcing mechanism which is the principal control mechanism used in continuous-time control theory, we are interested in establishing connections as well as contrasts between the two types of control theory; special attention will be given to several basic control-theoretic ideas including reachability, stability, and control-barriers \cite{amestabuada}.

\ifCLASSOPTIONcaptionsoff
  \newpage
\fi

\bibliographystyle{IEEEtran}
\bibliography{IEEEabrv,references}

\begin{thebibliography}{10}
\providecommand{\url}[1]{#1}
\csname url@samestyle\endcsname
\providecommand{\newblock}{\relax}
\providecommand{\bibinfo}[2]{#2}
\providecommand{\BIBentrySTDinterwordspacing}{\spaceskip=0pt\relax}
\providecommand{\BIBentryALTinterwordstretchfactor}{4}
\providecommand{\BIBentryALTinterwordspacing}{\spaceskip=\fontdimen2\font plus
\BIBentryALTinterwordstretchfactor\fontdimen3\font minus
  \fontdimen4\font\relax}
\providecommand{\BIBforeignlanguage}[2]{{%
\expandafter\ifx\csname l@#1\endcsname\relax
\typeout{** WARNING: IEEEtran.bst: No hyphenation pattern has been}%
\typeout{** loaded for the language `#1'. Using the pattern for}%
\typeout{** the default language instead.}%
\else
\language=\csname l@#1\endcsname
\fi
#2}}
\providecommand{\BIBdecl}{\relax}
\BIBdecl

\bibitem{bookcai}
W.~Wonham and K.~Cai, \emph{Supervisory Control of Discrete-Event
  Systems}.\hskip 1em plus 0.5em minus 0.4em\relax Springer, 2019.

\bibitem{cassandras2009introduction}
C.~G. Cassandras and S.~Lafortune, \emph{Introduction to discrete event
  systems}.\hskip 1em plus 0.5em minus 0.4em\relax Springer Science \& Business
  Media, 2009.

\bibitem{malik2002generating}
P.~Malik, ``Generating controllers from discrete-event models,'' in
  \emph{Proceedings of the Summer School in Modelling and Verification of
  Parallel processes}, 2002, pp. 337--242.

\bibitem{Reijnen19}
F.~Reijnen, A.~Hofkamp, J.~van~de Mortel-Fronczak, M.~Reniers, and J.~Rooda,
  ``Finite response and confluence of state-based supervisory controllers,'' in
  \emph{2019 IEEE 15th International Conference on Automation Science and
  Engineering (CASE)}, 2019, pp. 509--516.

\bibitem{Reijnen2022}
F.~F.~H. Reijnen, T.~R. Erens, J.~M. van~de Mortel-Fronczak, and J.~E. Rooda,
  ``Supervisory controller synthesis and implementation for safety plcs,''
  \emph{Discrete Event Dynamic Systems}, vol.~32, p. 115–141, 2022).

\bibitem{Golaszewski1987}
C.~H. Golaszewski and P.~J. Ramadge, ``Control of discrete event processes with
  forced events,'' \emph{26th IEEE Conference on Decision and Control},
  vol.~26, pp. 247--251, 1987.

\bibitem{brandin1994supervisory}
B.~A. Brandin and W.~M. Wonham, ``Supervisory control of timed discrete-event
  systems,'' \emph{IEEE Transactions on Automatic Control}, vol.~39, no.~2, pp.
  329--342, 1994.

\bibitem{ZhaCaiWon13}
R.~Zhang, K.~Cai, Y.~Gan, Z.~Wang, and W.~Wonham, ``Supervision localization of
  timed discrete-event systems,'' \emph{Automatica}, vol.~49, no.~9, pp.
  2786--2794, 2013.

\bibitem{TakaiUshio06}
S.~Takai and T.~Ushio, ``A new class of supervisors for timed discrete event
  systems,'' \emph{Discrete Event Dynamic Systems}, vol.~16, no.~2, pp.
  257--278, 2006.

\bibitem{rashidinejad2018}
A.~Rashidinejad, M.~Reniers, and L.~Feng, ``Supervisory control of timed
  discrete-event systems subject to communication delays and non-fifo
  observations,'' \emph{IFAC-PapersOnLine}, vol.~51, no.~7, pp. 456--463, 2018.

\bibitem{Rashidinejad2020a}
A.~Rashidinejad, P.~van~der Graaf, and M.~Reniers, ``Nonblocking supervisory
  control synthesis of timed automata using abstractions and forcible events,''
  in \emph{2020 16th International Conference on Control, Automation, Robotics
  and Vision (ICARCV)}, 2020, pp. 1--8.

\bibitem{Rashidinejad2020b}
\BIBentryALTinterwordspacing
A.~Rashidinejad, P.~{van der Graaf}, M.~Reniers, and M.~Fabian, ``Non-blocking
  supervisory control of timed automata using forcible events,''
  \emph{IFAC-PapersOnLine}, vol.~53, no.~4, pp. 356--362, 2020, 15th IFAC
  Workshop on Discrete Event Systems WODES 2020 — Rio de Janeiro, Brazil,
  11-13 November 2020. [Online]. Available:
  \url{https://www.sciencedirect.com/science/article/pii/S2405896321000756}
\BIBentrySTDinterwordspacing

\bibitem{Ushio2005}
\BIBentryALTinterwordspacing
T.~Ushio and S.~Takai, ``Control-invariance of hybrid systems with forcible
  events,'' \emph{Automatica}, vol.~41, no.~4, pp. 669--675, 2005. [Online].
  Available:
  \url{https://www.sciencedirect.com/science/article/pii/S0005109804003097}
\BIBentrySTDinterwordspacing

\bibitem{Huang2008}
J.~Huang and R.~Kumar, ``Directed control of discrete event systems for safety
  and nonblocking,'' \emph{IEEE Transactions on Automation Science and
  Engineering}, vol.~5, no.~4, pp. 620--629, 2008.

\bibitem{Balemi93}
S.~Balemi, G.~Hoffmann, P.~Gyugyi, H.~Wong-Toi, and G.~Franklin, ``Supervisory
  control of a rapid thermal multiprocessor,'' \emph{IEEE Transactions on
  Automatic Control}, vol.~38, no.~7, pp. 1040--1059, 1993.

\bibitem{RW}
P.~J. Ramadge and W.~M. Wonham, ``Supervisory control of a class of discrete
  event processes,'' \emph{SIAM Journal on Control and Optimization}, vol.~25,
  no.~1, pp. 206--230, 1987.

\bibitem{WR}
W.~M. Wonham and P.~J. Ramadge, ``On the supremal controllable sublanguage of a
  given language,'' \emph{SIAM Journal on Control and Optimization}, vol.~25,
  no.~3, pp. 637--659, 1987.

\bibitem{flordal2007compositional}
H.~Flordal, R.~Malik, M.~Fabian, and K.~{\AA}kesson, ``Compositional synthesis
  of maximally permissive supervisors using supervision equivalence,''
  \emph{Discrete Event Dynamic Systems}, vol.~17, no.~4, pp. 475--504, 2007.

\bibitem{ouedraogo2011nonblocking}
L.~Ouedraogo, R.~Kumar, R.~Malik, and K.~Akesson, ``Nonblocking and safe
  control of discrete-event systems modeled as extended finite automata,''
  \emph{IEEE Transactions on Automation Science and Engineering}, vol.~8,
  no.~3, pp. 560--569, 2011.

\bibitem{Thuijsman2020}
S.~Thuijsman and M.~Reniers, ``Transformational supervisor synthesis for
  evolving systems,'' \emph{IFAC-PapersOnLine}, vol.~53, no.~4, pp. 309--316,
  2020.

\bibitem{ramadge1989control}
P.~J. Ramadge and W.~M. Wonham, ``The control of discrete event systems,''
  \emph{Proceedings of the IEEE}, vol.~77, no.~1, pp. 81--98, 1989.

\bibitem{van2014cif}
D.~van Beek, W.~Fokkink, D.~Hendriks, A.~Hofkamp, J.~Markovski, J.~van~de
  Mortel-Fronczak, and M.~Reniers, ``{CIF} 3: Model-based engineering of
  supervisory controllers,'' in \emph{Tools and Algorithms for the Construction
  and Analysis of Systems}.\hskip 1em plus 0.5em minus 0.4em\relax Springer,
  2014, pp. 575--580.

\bibitem{ESCET2023}
W.~Fokkink, M.~Goorden, D.~Hendriks, B.~van Beek, A.~Hofkamp, F.~Reijnen,
  P.~Etman, L.~Moormann, A.~van~de Mortel-Fronczak, M.~Reniers, K.~Rooda,
  B.~van~der Sanden, R.~Schiffelers, S.~Thuijsman, J.~Verbakel, and H.~Vogel,
  ``Eclipse {ESCET}\texttrademark: The {E}clipse {S}upervisory {C}ontrol
  {E}ngineering {T}oolkit,'' in \emph{TACAS 2023}, 2023).

\bibitem{LinWon88}
F.~Lin and W.~Wonham, ``On observability of discrete-event systems,''
  \emph{Information Science}, vol.~44, no.~3, pp. 173--198, 1988.

\bibitem{bookcai2}
K.~Cai and W.~Wonham, \emph{Supervisor Localization: A Top-Down Approach to
  Distributed Control of Discrete-Event Systems}.\hskip 1em plus 0.5em minus
  0.4em\relax Springer, 2016.

\bibitem{amestabuada}
A.~Ames, S.~Coogan, M.~Egerstedt, G.~Notomista, K.~Sreenath, and P.~Tabuada,
  ``Control barrier functions: theory and applications,'' in \emph{2019 18th
  European Control Conference}, 2019, pp. 3420--3431.

\end{thebibliography}

\end{document}